\documentclass[copyright,creativecommons]{eptcs}
\providecommand{\event}{ICE'09} 
\providecommand{\volume}{??}
\providecommand{\anno}{2009}
\providecommand{\firstpage}{1}
\providecommand{\eid}{??}

\usepackage{amsmath}
\usepackage{amssymb}
\usepackage{amsthm}
\usepackage{amsfonts}
\usepackage{amstext}

\newtheorem{definition}{Definition}
\newtheorem{theorem}{Theorem}
\newtheorem{proposition}{Proposition}
\newtheorem{example}{Example}
\newtheorem{lemma}{Lemma}

\newcommand{\msg}[1]{\op{msg}(#1)}
\newcommand{\pos}[1]{\op{pos}(#1)}
\newcommand{\Neg}[1]{\op{neg}(#1)}

\usepackage{diagrams}
\diagramstyle[h=2mm,w=3mm,PostScript=dvips
]

\newcommand{\skel}{\ensuremath{\mathbb{A}}}

\newarrow{StrNext}      ===={=>}
\newarrow{DotStrNext}   ....{=>}
\newarrow{DotTo}   ....{>}
\newarrow{Bond} ----{->}
\newarrow{DashTo}{}{dash}{}{dash}>
\newarrow{MapTo}      |---{->}

\newarrow{StrNext}      ===={=>}
\newarrow{Bond} ----{->}
\newarrow{DashTo}{}{dash}{}{dash}>

{\bfseries}{\itshape}
{\bfseries}{\itshape}
{\bfseries}{\upshape}
{\bfseries}{\itshape}

\newtheorem{assumption}{Assumption}

\newcommand{\nodes}[1]{\ensuremath{\mathsf{nodes}(#1)}}

\newcommand{\TOP}[1]{{\sf top}(#1)}

\newcommand{\lts}[2]{\ \stackrel{#1}\longrightarrow_{#2}\ }
\newcommand{\ltsstar}[1]{\ \longrightarrow^*_{#1}\ }
\newcommand{\msgbox}[2]{[#1]_{#2}}
\newcommand{\interact}[4]{#1\rightarrow#2:\mathsf{#3}\langle#4\rangle}

\newcommand{\pp}{\mathrel{\boldsymbol{\mathord{\mid}}}}

\newcommand{\pfx}{\mathbf.\ }
\newcommand{\INACT}{\mathbf0 }

\newcommand{\who}{\mathsf{who}}

\newcommand{\Rule}[2]{\displaystyle{\frac{#1}{#2}}}
\newcommand{\Did}[1]{(\textsc{#1})}
\newcommand{\NI}{\noindent}

\newcommand{\op}[1]{{\sf #1}}

\newcommand{\absem}[1]{\{\!\!\!\{#1\}\!\!\!\}^{}}

\title{Choreographies with Secure Boxes and Compromised Principals}
\author{Marco Carbone\thanks{The first author is partially funded by
    the CosmoBiz project.  The second author is partially funded by
    the National Science Foundation, (grant no.~CNS-0952287).}
  \institute{IT University of Copenhagen\\
    Denmark} \email{carbonem@itu.dk} \and Joshua Guttman
  \institute{Worcester Polytechnic Institute\\
    United States} \email{guttman@wpi.edu} }
\def\titlerunning{Choreographies, Secure Boxes and Compromised
  Principals} \def\authorrunning{M. Carbone \& J. Guttman}
\begin{document}
\maketitle

\begin{abstract}
We equip choreography-level session descriptions with a simple
abstraction of a security infrastructure.  Message components may be
enclosed within (possibly nested) "boxes" annotated with the intended
source and destination of those components.  The boxes are to be
implemented with cryptography.

Strand spaces provide a semantics for these choreographies, in which
some roles may be played by compromised principals.  A \emph{skeleton}
is a partially ordered structure containing local behaviors (strands)
executed by \emph{regular} (non-compromised) principals.  A skeleton
is \emph{realized} if it contains enough regular strands so that it
could actually occur, in combination with any possible activity of
compromised principals.  It is \emph{delivery guaranteed (DG)
  realized} if, in addition, every message transmitted to a regular
participant is also delivered.

We define a novel transition system on skeletons, in which the steps
add regular strands.  These steps solve tests, i.e. parts of the
skeleton that could not occur without additional regular behavior.

We prove three main results about the transition system.  First, each
minimal DG realized skeleton is reachable, using the transition
system, from any skeleton it embeds.  Second, if no step is possible
from a skeleton $\skel$, then $\skel$ is DG realized.  Finally, if a
DG realized $\skel'$ is accessible from $\skel$, then $\skel'$ is
minimal.  Thus, the transition system provides a systematic way to
construct the possible behaviors of the choreography, in the presence
of compromised principals.
\end{abstract}

\section{Introduction}
\label{sec:introduction}
Distributed transactions are increasingly central to our economic and
social infrastructure.  Rigorous, type-based notions of session are
thus subjects of intense exploration, as they can ensure that
communications among principals are properly
coordinated~\cite{THK,honda.vasconcelos.kubo:language-primitives,HYC08,BCDDDY:concur2008,MostrousYoshidaHonda09}.
However, sessions require a security infrastructure, since the data
they carry may be sensitive, and a transaction may (for instance)
transfer money from one person to another.  Standard security
infrastructures, such as TLS~\cite{DierksEtAl99} for web interactions,
are two-party, point-to-point mechanisms.  When a transaction involves
more than two parties---for instance, a buyer, a seller, and a
bank---then it is hard to see how to use TLS sessions to ensure that
the parties get any security guarantees.

An alternative---given a session choreography---is to synthesize a
security infrastructure that is appropriate to the goals of that
session~\cite{CDFBL08,CDFBL09}.  This infrastructure is effectively a
custom cryptographic protocol generated specifically to ensure that
malicious principals cannot undermine the behavior that the advertised
session choreography promises to compliant principals.  Generating
this protocol, and ensuring its correctness, requires reasoning at
several levels, including both the choreography level and the
cryptographic level.

In this paper we study reasoning specifically at the choreography
level, without introducing the complexities of realistic cryptography.
These complexities include selection of public-key and symmetric
cryptographic primitives, as well as key distribution.  Another recent
paper which treats protocols by an abstraction of their cryptographic
mechanisms is~\cite{MaffeiEtAl07}.

We use a simple choreography-level specification for security of parts
of messages, which we call \emph{boxes}.  A box $[M]_{\rho_1\rho_2}$
represents the fact that message $M$ will be sent in some format $x$
such that, if $\rho_1$ and $\rho_2$ are uncompromised roles, then $x$
was prepared only by $\rho_1$ and can be opened only by $\rho_2$.
Boxes may appear nested inside other boxes.  Naturally, any
implementation of boxes will require cryptography.  We might implement
boxes by message structures in which $\rho_1,\rho_2$ first agree on a
shared secret, and then use it to encrypt and provide message
authentication for $M$ (and other messages as determined by the
choreography).  The first step of agreeing on a shared secret may rely
on public-key cryptography.  Boxes are a mechanism to specify when a
message component achieves secrecy and integrity between two
uncompromised principals, despite other compromised principals
behaving unpredictably or maliciously.

In this paper, we will develop a method to define the possible
behaviors of a choreography as a function of a choice of compromised
roles $R$.  That is, given an assumption that principals not in $R$
will behave in accordance with their roles in the choreography, we
would like to define all possible behaviors a choreography execution
can exhibit.  To do so, we translate each choreography description
into a set of strands.  Each of these strands represents a possible
local behavior of one principal in a single session, running a role of
the choreography.  These \emph{regular}, non-compromised strands may
interact with each other and with any behavior within the power of the
adversary, to produce a variety of global executions.  We give a
method for generating all of these global executions, or more
precisely, for finding the minimal, essentially different executions.  

We call these minimal, essentially different executions \emph{shapes}.
Each shape is a shape \emph{relative} to some starting point,
typically some assumed local strand representing a behavior of a
single participant.  The shapes describe the possible
\emph{explanations} for the experience of this participant, i.e.~what
other local executions (strands) of regular participants would be
needed in possible runs, in combination with adversary actions.  They
are minimal in that no lesser amount of regular behavior would yield a
full explanation of regular activity in the starting point.

We generate shapes via a transition system defined by two rules.  One
rule says that additional strands must be added when a participant
receives a box that the adversary could not create, and which is not
yet explained by an earlier transmission from an uncompromised
strand.  It also applies to situations where a box has been removed
from nested boxes, and only regular strands can extract it.  

The other rule corresponds to the usual choreography assumption on the
communication medium.  This assumption is that the medium is
\emph{resilient}, i.e.~that when an uncompromised participant sends a
message to another uncompromised participant, then that message will
be delivered.  Since we work in a partially ordered execution model,
there is no assumption about \emph{when} this message will be
delivered, relative to causally unrelated actions.  
We present three main results.  
\begin{enumerate}
  \item In the transition system defined by our two rules, and
  relative to a chosen assumption $R$ about compromised roles, if
  $\skel'$ is any shape compatible with a starting point $\skel$, then
  $\skel\longrightarrow^* \skel'$.  The same holds for shapes with
  guaranteed delivery.  (Thm.~\ref{thm:completeness}.)
  \item When we start from a single strand $\skel$, then any maximal
  trace $\skel\ltsstar S\skel'\not\rightarrow$ terminates with a shape
  $\skel'$ with delivery guaranteed.  (Thm.~\ref{thm:soundness}.)
  \item Every trace starting from a single strand terminates.  (Thm.~\ref{thm:termination}.)
\end{enumerate}
In particular, the first point holds for all strand spaces based on
boxes, while the second and third are specific to strand spaces
defined as the semantics of choreographies in a particular syntax.


\section{Abstract Strand Spaces}
\label{sec:abstractstrandspaces}
\subsection{Basic Definitions}
\begin{definition}[Messages and Boxes]
  Messages $M$ and boxes $b$ are defined:
  \begin{align*}
    W::= & \phantom{{}\mid\ {}} v\mid b
    &
    M::= & \phantom{{}\mid\ {}}\tilde W 
    &
    b::= & 
    \phantom{{}\mid\ {}}
    \msgbox{\tilde M}{\rho_1\rho_2}
  \end{align*}
  where $\tilde \cdot$ denotes a tuple of zero or more elements. $v$
  is a basic value---belonging to a \emph{finite} set of basic
  values---and $\rho_i$ ranges over the set of roles $\mathcal R$.
\end{definition}
\NI A message $M$ can either be a value $v$ or a box $\msgbox{\tilde
  M}{\rho_1\rho_2}$.  We also use letter $c$ to denote boxes.
A box is a tuple of messages $M_i$ from $\rho_1$ that can only be
opened by $\rho_2$. 

%
A {\em strand space}, first introduced in \cite{strandspaces} as a
formalism for reasoning about cryptographic protocols, is a collection
of strands. Here, we introduce {\em abstract strand spaces}, strand
spaces where messages range over $M$ (unlike the original version with
cryptography).  A \emph{substitution} is a function that maps basic
values to basic values.  Since basic values form a finite set, there
are only finitely many substitutions.
\begin{definition}[Abstract Strand Space]
  A {\em directed term} is a pair denoted by $\pm M$ where
  $\pm\in\{-,+\}$ is a direction with $+$ representing transmission
  and $-$ reception.  A \emph{trace} is an element of $(\pm M)^*$, the
  set of finite sequences of directed terms.    

  \NI An {\em abstract strand space} is a set $S$ with a trace mapping
  $\mathsf{tr}:S\rightarrow(\pm M)^*$.  A \emph{strand} is an element
  of $S$.

  \NI A strand space $S$ is \emph{closed under} a set of substitutions
  $\Sigma$, if, for every $s\in S$ and $\sigma\in\Sigma$, there is an
  $s'\in S$ such that $\mathsf{tr}(s')=\sigma(\mathsf{tr}(s'))$.

\end{definition}
In this paper we consider \emph{finite} strand spaces that are closed
under substitutions of basic values for basic values.

\NI {\bf Notation.} If $s\in S$ is a strand then $s(i)$ denotes the
$i^{\text{th}}$ element of the trace of $s$ and is called {\em
  node}. We write $m\Rightarrow n$ when $n$ is the node immediately
after $m$ on the same strand $s$ i.e. $m=s(i)$ and $n=s(i+1)$. Also,
$\msg n$ denotes the message of the directed term in $n$ while $\Neg
n$ ($\pos n$) holds if $n$ is a reception (transmission) node.

\NI It is now interesting to see how these input/output traces could
be combined together in order to form a real execution. {\em
  Skeletons} express parts of an execution (with some pending
transmission/reception nodes related to adversary activity):
\begin{definition}[Skeleton]
  Given 
  a strand space $S$, 
  a \emph{skeleton} $\mathbb A$ is a finite set of regular nodes
  (nodes belonging to strands of $S$), denoted by
  $\mathsf{nodes}(\mathbb{A})$, equipped with a partial order
  $\preceq_{\mathbb A}$ on $\mathsf{nodes}(\mathbb{A})$ indicating
  causal precedence (consistent with $\Rightarrow$). Moreover, if
  $m\Rightarrow n$ and $n\in\nodes \skel$, then $m\in\nodes \skel$.
\end{definition}

\smallskip

\NI In the rest of the paper, $\prec$ will denote the non-reflexive
subrelation of $\preceq$.

\smallskip

\begin{example}\rm
  As an example, let us consider a skeleton composed by three
  strands. Below, outgoing and incoming edges denote transmission and
  reception nodes respectively.  {\small
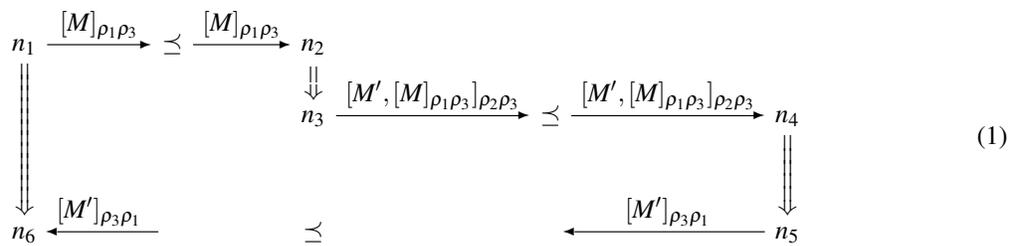
    \begin{center}
      \begin{equation}\label{skeleton1}
        \begin{diagram}
          n_1      &\rTo^{\ \msgbox{M}{\rho_1\rho_3}\ } & \preceq & \rTo^{\ \msgbox{M}{\rho_1\rho_3}\ } & n_2\\\\
          &&&&\dImplies \\
          &&&& n_3  & 
          \rTo^{\ \msgbox{M',\msgbox{M}{\rho_1\rho_3}}{\rho_2\rho_3}\ } & \preceq & \rTo^{\ \msgbox{M',\msgbox{M}{\rho_1\rho_3}}{\rho_2\rho_3}\ }& n_4\\\\\\
          \dImplies &&&&&&&& \dImplies\\\\\\
          n_6 &\lTo^{\ {\msgbox{M'}{\rho_3\rho_1}}\ }
          &\ &\ &\preceq&\ &\ &
          \lTo^{\ {\msgbox{M'}{\rho_3\rho_1}}\ }
          & n_5
        \end{diagram}    
  \end{equation}
\end{center}} \NI The three strands above belong to roles $\rho_1$,
$\rho_2$ and $\rho_3$ respectively. If the middle strand was not there
e.g. if $\rho_2$ were compromised, then we would have the following
skeleton: {\small
  \begin{center}
    \begin{diagram}
      n_1  &\rTo^{\ \msgbox{M}{\rho_1\rho_3}\ }   &\ \preceq\ & \rTo^{\ \msgbox{M',\msgbox{M}{\rho_1\rho_3}}{\rho_2\rho_3}\ }& n_4\\\\\\
      \dImplies &                        &           &             & \dImplies\\\\\\
      n_6  &\lTo^{\ {\msgbox{M'}{\rho_3\rho_1}}\ }&\ \preceq\ &    \lTo^{\ {\msgbox{M'}{\rho_3\rho_1}}\ }       & n_5
    \end{diagram}    
  \end{center}} 
\end{example}


As previously said, some roles may belong to compromised
principals. In the sequel, we set $R\subseteq\mathcal R$ to be the set
of compromised roles. Moreover, we assume that each strand is always
marked with the role it belongs to (a strand belongs to exactly one
role). On this premises, it is natural to define the untamed behaviour
of $R$ (or adversary) in terms of {\em penetrator} strands:
\begin{definition}[Abstract Penetrator]
  $\mathcal{P}_R$, the \emph{abstract penetrator} for a set of
  compromised roles $R$, is the set of strands of the forms:
  \begin{center}
    \begin{tabular}{lllllllllllllll}
      (\textbf{C})   &  $-M_0\Rightarrow -(M_1,\ldots,M_k)\Rightarrow +(M_0,M_1,\ldots,M_k)$
      &\qquad\qquad
      (\textbf{A})   &  $+\tilde v$       
      \\[1mm]
      (\textbf{S})   &  $-(M_0,M_1,\ldots,M_k)\Rightarrow +M_0\Rightarrow+(M_1,\ldots,M_k)$
      &\qquad\qquad
      \\[1mm]
      (\textbf{B})   &  $-\tilde{M}\Rightarrow+[\tilde{M}]_{\rho_1\rho_2}$\quad where $\rho_1\in R$
      &\qquad\qquad
      \\[1mm]
      (\textbf{O})   &  $-[\tilde{M}]_{\rho_1\rho_2}\Rightarrow+\tilde{M}$\quad where $\rho_2\in R$
    \end{tabular}
  \end{center}
\end{definition}
Above, (\textbf{C}) allows the penetrator to {\em compose} received
messages and resend them; in (\textbf{S}), a compound message can be
{\em separated} and resent; (\textbf{B}) allows to {\em box} messages
and sign them with a compromised role (from $R$); with (\textbf{O}),
the penetrator can {\em open} boxes targeted to compromised roles; and
using (\textbf{A}), the penetrator can send any clear text. 

\smallskip 

We can compose (instances of) the various strands above with a
skeleton in order to build the {\em graph of interaction} of a
skeleton $\skel$ with respect to a penetrator $\mathcal{P}_R$ i.e. an
acyclic directed graph $\mathcal{B}$ whose nodes are the nodes of
strands in $\mathcal{P}_R$ and $\skel$, and whose edges can be
obtained by connecting any transmitting node $m$ with a receiving node
$n$ such that $\msg m=\msg n$.  We say of two nodes $m_0,m_k$ of
$\mathcal{B}$ that $m_0\preceq_{\mathcal{B}} m_k$ if there is a
sequence $m_0,m_1,\ldots,m_k$ such that for each pair $m_i,m_{i+1}$,
either $m_i\preceq_{\skel}m_{i+1}$, or $m_i\Rightarrow^+ m_{i+1}$ on a
penetrator strand of $\mathcal{P}_R$, or $m_i$ is a transmission node
and $ m_{i+1}$ is a receiving node connected to it.

\smallskip

\NI We now define a \emph{realized skeleton} i.e. a skeleton that has
precisely the behavior of some execution:
\begin{definition}[Realized Skeleton]
  A skeleton $\skel$ is \emph{realized} if there is a graph of
  interaction $\mathcal{B}$ of $\skel$ wrt $\mathcal{P}_R$ such that
  every reception node has an incoming edge, and for all nodes
  $m,n\in\nodes\skel$, $m\preceq_{\mathcal{B}} n$ implies
  $m\preceq_{\skel} n$.
\end{definition}

\smallskip

A {\em shape} is a minimal homomorphism preserving $\preceq$ that maps
a skeleton into a realized one. Below, a homomorphism
$H:\skel_0\mapsto\skel_1$ is node-wise injective if it is an injective
function on the nodes of $\skel_0$. Moreover, $H_0$ is node-wise less
than or equal to $H_1$, written $H_0\leq H_1$, if for some node-wise
injective $L$, $L\circ H_0=H_1$. We then say that $H_0$ is node-wise
minimal in some set $Z$ whenever $H_0\in Z$ and for any $H\in Z$,
$H\leq H_0$ implies $H$ and $H_0$ are isomorphic.
\begin{definition}[Shape \cite{DGT07}]
  $H:\skel_0\mapsto\skel'$ is a shape for $\skel_0$ if $H$ is
  node-wise minimal among the set of homomorphisms
  $H':\skel_0\mapsto\skel''$ where $\skel''$ is realized.
\end{definition}

\smallskip

\NI Sometimes, with an abuse of terminology, if $H:\skel\mapsto\skel'$
is a shape for $\skel$, we shall say that $\skel'$ is a shape for
$\skel$.  For instance, the skeleton in (\ref{skeleton1}) is a shape
for $n_1\Rightarrow n_6$. On the other hand, because of the extra node
$n_4$, the following realized skeleton is not a shape
for $n_1\Rightarrow n_6$:\\[-7mm]
{\small
\begin{center}
  \begin{equation}\label{skeleton2}
    \begin{diagram}
      n_1      &\rTo^{\ \msgbox{M}{\rho_1\rho_3}\ } & \preceq &
      \rTo^{\ \msgbox{M}{\rho_1\rho_3}\ } & n_2
      & \qquad & \lTo^{\ \msgbox{v}{\rho_2\rho_3}\ }& n_4\\\\\\
               \dImplies&&&&\dImplies \\
       n_6 &\lTo^{\ \msgbox{\msgbox{M}{\rho_3\rho_1}}{\rho_3\rho_1}\ }
       &\preceq & \lTo^{\ \msgbox{\msgbox{M}{\rho_3\rho_1}}{\rho_3\rho_1}\ } &n_3&\ &
    \end{diagram}    
  \end{equation}
  \end{center}}
  We also consider special skeletons which guarantee that messages are
  delivered.
\begin{definition}[Delivery-Guaranteed Skeletons]
  A {\em delivery-guaranteed skeleton} (DG skeleton) is a skeleton
  such that for every positive node $n$ such $\msg n=\msgbox{\tilde
    M}{\rho_1\rho_2}$ and $\rho_2\not\in R$ there exists a negative
  node $n'$ on another strand such that $\msg n=\msg n'$.
\end{definition}
\NI Note that (\ref{skeleton2}) is not DG while (\ref{skeleton1}) is.
Delivery-guaranteed skeletons characterize some special shapes: 
\begin{definition}[Delivery-Guaranteed
  Shape]\label{definition:dgshape}
  $H:\skel_0\mapsto\skel'$ is a DG shape for $\skel_0$ if $H$ is
  node-wise minimal among the set of homomorphisms
  $H':\skel_0\mapsto\skel''$ where $\skel''$ is a realized and DG
  skeleton.
  
\end{definition}


\subsection{Characterizing Realized Skeletons}
In this subsection, we will introduce a characterization of realized
skeletons in the spirit of \cite{G09}. The idea is to use {\em
  authentication tests} \cite{DGT07} as a method for explaining why a
message is suddenly found outside a box which was previously
containing it. In general, either the box owner is compromised or else
it was transmitted by a regular strand. The following definition
formalizes the idea of a message occurring inside or outside a set of
boxes.
\begin{definition}
  A message $M_0$ is {\em found only within a set of boxes $B$ in
    $M_1$}, written $M_0\odot^BM_1$, whenever every occurrence of
  $M_0$ in $M_1$
  is nested inside a box of $B$.\\
  A message $M_0$ is {\em found outside $B$ in $M_1$}, written
  $M_0\dagger^BM_1$, whenever not $M_0\odot^BM_1$.
\end{definition}
\NI As an example, for $M\not="Hi"$, $M$ is found only within
$\{\msgbox{M}{\rho_1\rho_2},\,\msgbox{\msgbox{M,"Hi"}{\rho_3\rho_1},"Hi"}{\rho_3\rho_4}\}$
in $\msgbox{\msgbox{M}{\rho_1\rho_2}}{\rho_2\rho_3}$ and
$\msgbox{\msgbox{\msgbox{(M,"Hi")}{\rho_3\rho_1},"Hi"}{\rho_3\rho_4}}{\rho_4\rho_1}$. Also,
$M$ is found only within $\{\msgbox{"Hi"}{\rho_3\rho_1}\}$ in
$\msgbox{\msgbox{"Hi"}{\rho_1\rho2}}{\rho_3\rho_1}$ as it does not
occur at all. On the contrary, $M$ is found outside
$\{\msgbox{"Hi"}{\rho_1\rho_2}\}$ in
$\msgbox{M,\msgbox{"Hi"}{\rho_1\rho_2}}{\rho_3\rho_4}$.

Given a skeleton, a set of boxes and a message, we can highlight those
minimal nodes for which such a message is found only outside the
boxes. This is formalized by the notion of cut:
\begin{definition}[Cut]
  Let $M$ be a message, $B$ a set of boxes and $\mathbb A$ a
  skeleton. Then,
  \[\textsf{Cut}(M,B,\mathbb A)= \{n\in\textsf{nodes}(\mathbb
  A):\exists m\preceq_{\mathbb A}n\text{ and }
  M\dagger^B\textsf{msg}(m)\}\] $\textsf{Cut}(M,B,\mathbb A)$ is
  defined whenever there exists a node $n$ in $\mathbb A$ such that
  $M\dagger^B\textsf{msg}(n)$.
\end{definition}

\smallskip

\NI Note that $M$ occurs outside $B$ in all minimal nodes of
$\textsf{Cut}(M,B,\mathbb A)$. In the following skeleton $\skel$,
{\small
  \begin{center}   
    \begin{equation}\label{skeleton3}
    \begin{diagram}
      n_1      &\rTo^{\ \msgbox{\msgbox{M}{\rho_1\rho_3}}{\rho_1\rho_2}\ } & \ \preceq\ & \rTo^{\ \msgbox{\msgbox{M}{\rho_1\rho_3}}{\rho_1\rho_2}\ } & n_2\\\\
      &                &       &                & \dImplies \\
      &&&& n_3   & \rTo^{\
        \msgbox{M',\msgbox{M}{\rho_1\rho_3}}{\rho_2\rho_3}\ } & \ \preceq\ & \rTo^{\ \msgbox{M',\msgbox{M}{\rho_1\rho_3}}{\rho_2\rho_3}\ }& n_4\\\\\\
      \dImplies&                &       &                &  &                &       &                & \dImplies\\
      n_6 & \lTo^{\ \msgbox{\msgbox{M}{\rho_1\rho_3}}{\rho_2\rho_1}\ } & \quad &&&& \quad & \lTo^{\ \msgbox{\msgbox{M}{\rho_1\rho_3}}{\rho_3\rho_2}\ }& n_5
    \end{diagram}    
    \end{equation}
  \end{center}}
\NI $\textsf{Cut}(\msgbox{M}{\rho_1\rho_3},B,\mathbb A)$ is the set
$\{n_3,n_4,n_5,n_6\}$ with minimal nodes $\{n_3,n_6\}$ for
$B=\{\msgbox{\msgbox{M}{\rho_1\rho_3}}{\rho_1\rho_2}\}$ and the whole
skeleton $\skel$ for $B=\emptyset$ (assuming that $\{n_1\}\in\rho_1$,
$\{n_2,n_3,n_6\}\in\rho_2$ and $\{n_4,n_5\}\in\rho_3$ and no role is
in $R$). Also, $\textsf{Cut}(\msgbox{M}{\rho_1\rho_3},B,\mathbb
A)=\{n_6\}$ for
$B=\{\msgbox{\msgbox{M}{\rho_1\rho_3}}{\rho_3\rho_2}\}$ but empty if
$\rho_2$ were compromised.  In the subskeleton $\skel'$ composed by
nodes $n_1$ and $n_2$ we have
$\textsf{Cut}(\msgbox{M}{\rho_1\rho_3},\{\msgbox{\msgbox{M}{\rho_1\rho_3}}{\rho_1\rho_2}\},\skel')=\emptyset$.

  The idea behind authentication tests is that any minimal node in a
  cut needs to be explained in the skeleton. In other words, there
  must be an earlier sequence of events that extracted the message out
  of some box or {\em legally} created it. Formally,
\begin{definition}[Solved Cut]\label{definition:solvedcut}
  A cut \textsf{Cut}$(M,B,\mathbb A)$ is solved wrt a set of
  compromised roles $R$, if for any of its $\prec_{\mathbb A}$-minimal
  nodes $m_1$:
  \begin{enumerate}
  \item either $m_1$ is a transmission node;
  \item or $M=\msgbox{\tilde M}{\rho_1\rho_2}$ and $\rho_1\in R$, or
    for some $\msgbox{\tilde M}{\rho_1\rho_2}\in B$, $\rho_2\in R$.
  \end{enumerate}
\end{definition}
\NI The definition above says that a cut \textsf{Cut}$(M,B,\mathbb A)$
is solved whenever, for every minimal reception node $n$, $M$ is
outside $B$ in $n$ because of some penetrator activity. For instance,
in (\ref{skeleton3}), $\textsf{Cut}(\msgbox{M}{\rho_1\rho_3},B,\mathbb
A)=\{n_6\}$ is not solved for
$B=\{\msgbox{\msgbox{M}{\rho_1\rho_3}}{\rho_3\rho_2},
\msgbox{\msgbox{M}{\rho_1\rho_3}}{\rho_1\rho_2},
\msgbox{M,\msgbox{M}{\rho_1\rho_3}}{\rho_2\rho_3}\}$ and $R=\emptyset$
while it is solved if $R=\{\rho_2\}$.  The above definition turns to
be a crucial property of realized skeletons. In fact, the following
proposition states that the property of being realized is
characterized by all if its cuts being solved.

\smallskip

\begin{proposition}\label{theorem:cutrealized}
  Let $\mathbb A$ be a skeleton. Then, every cut in $\mathbb A$ is
  solved if and only if  $\mathbb A$ is realized.
\end{proposition}

\begin{proof}
$\Rightarrow$:\quad

We prove this by contradiction. Assume that $\mathbb A$ is not
realized. Then, by definition, there must be an input node $n$
containing a message that a penetrator $\mathcal P_R$ is not allowed
to send i.e. there is some node $n$ such that, for all $m\prec_\skel
n$, either (i) $\rho_1\not\in R$, $\msgbox{\tilde M}{\rho_1\rho_2}$ is
nested in $\msg n$ and does not occur in $\msg{m}$; or (ii) for some
message $M$ and $\rho_2\not\in R$, we have that
$M\dagger^{\{\msgbox{\tilde M}{\rho_1\rho_2}\}}\msg n$ and
$M\odot^{\{\msgbox{\tilde M}{\rho_1\rho_2}\}}\msg m$.
    If (i) holds, then $\textsf{Cut}(\msgbox{\tilde
      M}{\rho_1\rho_2},\emptyset,\skel)$ is clearly
    unsolved. Similarly, if (ii) then $\textsf{Cut}(M,\{\msgbox{\tilde
      M}{\rho_1\rho_2}\},\skel)$ is unsolved.

$\Leftarrow$:\quad
Assume that there is a cut
    \textsf{Cut}$(M,B,\mathbb A)$ which is not solved. That means,
    that there is a minimal input node where $M$ is only found inside
    $B$ and such that $M\not=\msgbox{\tilde M}{\rho_1\rho_2}$ for
    $\rho_1\in R$, and for no $\msgbox{\tilde M}{\rho_1\rho_2}\in B$,
    $\rho_2\in R$. But then, there is no penetrator activity which
    could derive $M$ hence $\skel$ would not be realized.
\end{proof}
We conclude this section observing that an unsolved cut implies the
existence of another unsolved cut whose boxes $B$ are messages
appearing in the current skeleton. In the sequel, let the relation $M
\sqsubset M'$ hold whenever $M$ is contained in $M'$ ($\sqsubseteq$ is
the reflexive closure).


\begin{proposition}\label{theorem:minimalB}
  Let $\mathbb A$ be a skeleton and
  $\textsf{Cut}(\msgbox{M}{\rho_1\rho_2},B,\mathbb A)$ an unsolved
  cut. Then, $\textsf{Cut}(\msgbox{M}{\rho_1\rho_2},B',\mathbb A)$ is
  also unsolved for $B'=\{b \pp n'\prec_\skel n\;\text{ s.t. }\
  \msgbox{\tilde M}{\rho_1\rho_2} \sqsubset b\sqsubseteq
  \textsf{msg}(n') \land \textsf{rcv}(b)\not\in R \ \}$ for some
  $\preceq_{\skel}$-minimal input node $n$ in
  $\textsf{Cut}(\msgbox{M}{\rho_1,\rho_2},B,\mathbb A)$ and
  $\rho_1\not\in R$.
\end{proposition}

\begin{proof}
  Let $c=\msgbox{\tilde M}{\rho_1\rho_2}$. From
  Definition~\ref{definition:solvedcut}, there exists a
  $\preceq_{\skel}$-minimal {\em input} node $n$ in $\textsf{Cut}(c,
  B, \skel)$, such that $\rho_1\not\in R$ and $\rho_4\not\in R$ for
  all $\msgbox{\tilde M}{\rho_3\rho_4}\in B$.
  Let us now consider the predecessors $n'$ of $n$ in $\skel$ which,
  by definition of cut, are such that $c\odot^B\textsf{msg}(n')$. We
  consider two cases:
    (i) if none of $n$'s predecessors contains $c$ then
    $B'=\emptyset$ and therefore $\textsf{Cut}(c, \emptyset, \skel)$
    is unsolved as $n$ is a minimal node such that
    $c\dagger^\emptyset\textsf{msg}(n')$;
    (ii) $n_1\ldots n_k$ are $n$'s predecessors such that
    $c\sqsubset\textsf{msg}(n_i)$.  Letting
    $$ B_i=\{b\pp c\sqsubset b\sqsubseteq
    \textsf{msg}(n_i)\text{ and }\textsf{rec}(b)\not\in R\},$$
    $B_i\subseteq B$, because $n$ is a minimal node in
    $\textsf{Cut}(c, B, \skel)$.  Thus, as
    $c\odot^{\bigcup_iB_i}\textsf{msg}(n_i)$, $n$ is also minimal in
    $\textsf{Cut}(c, \bigcup_iB_i, \skel)$.  Finally, as $B'=\bigcup_i
    B_i$, we can conclude that $\textsf{Cut}(c,B',\mathbb A)$ is also
    unsolved.
\end{proof}


\section{Searching for Shapes}
\label{sec:shapesearching}
The results on cuts suggest a possible way of adding nodes to a
skeleton so that it can become realized. We shall now address this
problem and introduce a constructive method for deriving realized
skeletons from non-realized ones.
%
In the sequel, the operation $\skel\ \cup\ \uparrow_m\mbox{\em with }\
m\prec n$, returns the skeleton $\skel'$ consisting of $\skel$ and the
nodes $\{m'\pp m'\Rightarrow^* m\land m'\in \nodes{S}\}$, with the
ordering strengthened so that $m\prec_{\skel'}n$.  Similarly, $\skel\
\cup\ \uparrow_m\mbox{ with }\ n\prec m$ is the corresponding $\skel'$
with the opposite order enrichment $n\prec_{\skel'}m$.
\begin{definition}[Reduction Rules]
  The relation between skeletons $\skel\lts{} S\skel'$, is the
  minimum relation satisfying the following rules:
  \begin{center}
    \begin{equation*}
      \begin{array}{lll}
        \Did{A1}\quad&
        \Rule
        {
          \begin{array}{c}
            n\in\mathbb A\ \land\ \Neg n 
            \qquad\qquad 
            c=\msgbox{\tilde M}{\rho_1\rho_2}
            \qquad\qquad
            c \dagger^B\textsf{msg}(n)
            \\[1.5mm]
            m\in S\backslash\mathbb A\ \land\ \textsf{pos}(m)
            \qquad\quad
            \rho_1\not\in R
            \qquad\qquad
            c\dagger^B\textsf{msg}(m) 
            \\[1.5mm]
            \forall m'\,.\, m'\prec_{\mathbb A}n \lor m'\Rightarrow^+ m
            \quad\mbox{implies }\quad 
            c\odot^B\textsf{msg}(m')
          \end{array}
        }
        {
          \skel\lts{}S\skel\ \cup\ \uparrow_m\mbox{ with }\ m\prec n
        }
        \\\\
        \Did{A2}\quad&
        \Rule
        {
          n\in\mathbb A\ \land\ \pos n
          \quad
          \begin{array}{ll}
            \msg n=\msgbox{\tilde M}{\rho_1\rho_2}
            &\quad
            \lnot(\exists m'\in\textsf{neg}(\skel)\pfx n\prec m'
            \land
            \textsf{msg}(m')=\textsf{msg}(n))\\
            \rho_2\not\in R
            &\quad
            m\in(S)
            \ \land\
            \Neg m
            \ \land\
            \textsf{msg}(m)=\textsf{msg}(n)
          \end{array}
        }
        {
          \skel\lts{}S\skel\ \cup\ \uparrow_m\mbox{ with }\ n\prec m
        }
      \end{array}
    \end{equation*}
  \end{center}
  where the set of strands $S$ (strand space domain) is the set of
  regular strands.  Observe in rule (A1) that if there is any $B$ that
  satisfies the premise, then
$$B=\{b \pp n'\prec_\skel n\;\text{ s.t. }\ \msgbox{\tilde
  M}{\rho_1\rho_2}
\sqsubset b\sqsubseteq \textsf{msg}(n') \land \textsf{rcv}(b)\not\in R
\ \}$$
%
\end{definition}
We briefly comment the rules above. The first rule adds, when
possible, nodes that explain why a message is found outside a
box. Given a box $c$, the set of boxes $B$ and a node $n$ which is
minimal in \textsf{Cut}$(c,B,\skel)$, we choose $m$ to be the minimal
node preceding $n$ such that $c$ is found outside $B$. Note that $m$
may already be in the skeleton (added together with some $m'$ such
that $m\Rightarrow^* m'$) and the rule still be applicable because
$\preceq$ needs to be updated. The second rule deals with adding a
recipient, if any is found, to an output node. 
\begin{proposition}\label{proposition:reducesunrealized}
  If $\skel\lts{}{S}\skel'$ then $\skel$ is not realized or $\skel$ is
  not DG.
\end{proposition}
\begin{proof}
  If the reduction $\skel\lts{}{S}\skel'$ is obtained by applying rule
  \Did{A1}, then the cut $\textsf{Cut}(\msgbox{\tilde
    M}{\rho_1\rho_2},B,\skel)$ is clearly not solved. On the other
  hand, if $\skel\lts{}{S}\skel'$ by \Did{A2}, then we are clearly
  adding an input node to a pending output.
\end{proof}
In the sequel, we say that a homomorphism $H:\skel\mapsto\skel'$ is an
{\em augmentation} whenever $H$ is an inclusion (identity on the
domain $\skel$), any node in $\skel'\backslash\skel$ belongs to the
same strand and $\preceq_{\skel'}$ is an extension of $\preceq_\skel$.
Directly from the rules, it follows that:
\begin{proposition}\label{proposition:reductionaugmentation}
  Let $H$ map $\skel$ to $\skel'$ such that $\skel\lts{}S\skel'$. Then
  $H$ is an augmentation.
\end{proposition}
Building on the above proposition, we say that $H$ is of {\em type 1}
({\em type 2}) if it corresponds to the application of a rule 1 (rule
2).

\smallskip

In the sequel $\mathbb A\ltsstar S\mathbb A'$ holds whenever there
exists a finite sequence $\mathbb A_1\lts{}S\ldots\lts{}S\mathbb A_k$
such that $\mathbb A=\mathbb A_1$ and $\mathbb A'=\mathbb
A_k$. Moreover, $\mathbb A\not\rightarrow$ whenever there is no
$\mathbb A'$ such that $\mathbb A\lts{}S\mathbb A'$.  The following
result states that we can always reach all the shapes by repeatedly
applying the rules.
\begin{theorem}[Completeness]\ 
\label{thm:completeness}
  \begin{enumerate}

  \item Let $\skel$ be a single-strand skeleton and $H$ a shape such
    that $\skel\mapsto\skel'$.  Then $\skel\ltsstar S\skel'$ (up-to
    isomorphism).

  \item Let $\skel$ be a single-strand skeleton and $H$ a DG shape
    such that $\skel\mapsto\skel'$.  Then $\skel\ltsstar S\skel'$
    (up-to isomorphism).

  \end{enumerate}
\end{theorem}
\begin{proof}
  From Proposition~\ref{proposition:reductionaugmentation}, we only
  have to prove that shapes can be expressed as the composition of
  augmentations of type 1 or 2 (type 2 is only considered when proving
  point 2).  Formally, we show that there exist a $k$ such that for
  every $i\in\{0,1\ldots,k\}$ we have $H = L_i\circ
  H_i\circ\ldots\circ H_1\circ H_0$ where $L_i$ is a node-wise
  injective homomorphism, $H_0$ the identity mapping and
  $H_1,\ldots,H_i$ augmentations.

  The first step is to show how we can find $k$ and inductively
  construct each $H_i$ and $L_i$ starting from the identity:
  \begin{itemize}

  \item {\bf Base Case.} As $H_0$ must be the identity, we chose $L_0
    = H$ noting that $H$ is node-wise injective by definition of
    shape. We then have that $H = L_0\circ H_0$.

  \item {\bf Inductive Case.} Let $H = L_i\circ H_j\circ\ldots\circ
    H_1\circ H_0$ such that $H_0$ is the identity and $H_1,\ldots,H_i$
    are augmentations. If $L_i$ is an isomorphism then $i=k$ and we
    can stop. In fact, by definition of shape, $H$ is the minimum
    realized skeleton hence the image of $H_j\circ\ldots\circ H_1\circ
    H_0$ is isomorphic to $\skel'$, image of $H$.

    Let $L_i$ be not an isomorphism.  Moreover, let
    $L_i:\skel_j\mapsto\skel'$ and $H_j\circ\ldots\circ H_1\circ
    H_0:\skel\mapsto\skel_j$ for some $\skel_j$. We show how to
    construct $H_{i+1}$ and $L_{i+1}$. By definition of shape, as
    $L_i$ is not an isomorphism, $\skel_j$ is not realized. If that is
    the case, then either there is a dangling output (this is to be
    considered only when proving point 2) or, by
    Proposition~\ref{theorem:cutrealized}, there exists an unsolved cut
    $\mathsf{Cut}(\msgbox{\tilde M}{\rho_1\rho_2},B,\skel_j)$ i.e., by
    definition of cut, there exists an input node $m_1$,
    $\preceq_{\skel_j}$-minimal in $\textsf{Cut}(\msgbox
    M{\rho1\rho2}, B, \skel')$, such that $\rho_1\not\in R$ and for
    all $\msgbox{\tilde M}{\rho_3\rho_4}\in B$, $\rho_4\not\in
    R$. Now, as $\skel'$ is realized, all cuts must be solved. Then,
    because $L_i$ is node-wise injective, we can choose a node in the
    pre-image of $L_i$ which is not in $\skel_i$ but solves
    $\mathsf{Cut}(\msgbox{\tilde M}{\rho_1\rho_2},B,\skel_j)$ (or add
    the corresponding input when proving point 2). Adding this node,
    precisely corresponds to an augmentation induced by rule \Did{A1}
    (or \Did{A2}) which will be our $H_{i+1}$. We can then choose
    $L_{i+1}$ to be equal to $L_i$ but also mapping the new added node
    to $\skel'$ accordingly.
  \end{itemize}

  The above procedure shows how to construct the various $H_i$ and
  $L_i$. In order to complete the proof, we need to show that we
  always find the $k$. But this follows by the fact that augmentations
  always increase the size of a skeleton and observing that we stop
  once we reach an isomorphism.
\end{proof}

\begin{example}\rm
  Let $S=\{s_i\}_{i=1,\ldots,5}$, $s_1,s_2\in \rho_1$,
  $s_3,s_4\in\rho_2$, $s_5\in\rho_3$ and such that:
  {\small
  \begin{align*}
    s_1\ =\ &
    +\msgbox{\msgbox{\textsf{secret}}{\rho_1\rho_3}}{\rho_1\rho_2}
    \ \Rightarrow\ 
    -\msgbox{\textsf{reject}}{\rho_2\rho_1}
    \\
    s_2\ =\ &
    +\msgbox{\msgbox{\textsf{secret}}{\rho_1\rho_3}}{\rho_1\rho_2}
    \ \Rightarrow\ 
    -\msgbox{\msgbox{\textsf{newsecret}}{\rho_3\rho_1}}{\rho_2\rho_1}
    \\
    s_3\ =\ &
    -\msgbox{\msgbox{\textsf{secret}}{\rho_1\rho_3}}{\rho_1\rho_2}
    \ \Rightarrow\ 
    +\msgbox{\textsf{reject}}{\rho_2\rho_1}
    \\
    s_4\ =\ &
    -\msgbox{\msgbox{\textsf{secret}}{\rho_1\rho_3}}{\rho_1\rho_2}
    \ \Rightarrow\ 
    +\msgbox{\msgbox{\textsf{secret}}{\rho_1\rho_3}}{\rho_2\rho_3}
    \ \Rightarrow\ 
    -\msgbox{\msgbox{\textsf{newsecret}}{\rho_3\rho_1}}{\rho_3\rho_2}
    \ \Rightarrow\ 
    +\msgbox{\msgbox{\textsf{newsecret}}{\rho_3\rho_1}}{\rho_2\rho_1}
    \\
    s_5\ =\ &
    -\msgbox{\msgbox{\textsf{secret}}{\rho_1\rho_3}}{\rho_2\rho_3}
    \ \Rightarrow\ 
    +\msgbox{\msgbox{\textsf{newsecret}}{\rho_3\rho_1}}{\rho_3\rho_2}
  \end{align*}
  }
  If, for instance, $\rho_2\in R$ and we start from $s_5$, we can then
  apply \Did{A1} for $B=\emptyset$,
  $M=\msgbox{\textsf{secret}}{\rho_1\rho_3}$ and $m$ being the first
  node of the strands $s_1/s_2$. We obtain the following skeleton:
  {\small
  \begin{center}
  \begin{equation}\label{skeleton4}
    \begin{diagram}
                   &\bullet  & \rTo^{\ \msgbox{\msgbox{\textsf{secret}}{\rho_1\rho_3}}{\rho_1\rho_2}\ } 
                   & \preceq & 
                   \rTo^{\ \msgbox{\msgbox{\textsf{secret}}{\rho_1\rho_3}}{\rho_2\rho_3}\ } 
                   & \bullet\\
                   &         &&&&\dImplies\\\\\\
                &\       & \
                   & \       & 
                   \lTo^{\ \msgbox{\msgbox{\textsf{newsecret}}{\rho_3\rho_1}}{\rho_3\rho_2}\ }
                  & \bullet
    \end{diagram}    
  \end{equation}
  \end{center}
  }
  which is a shape for $s_5$. If we start from $s_2$, we can then
  apply \Did{A1} for $B=\emptyset$,
  $M=\msgbox{\textsf{secret}}{\rho_1\rho_3}$ and $m$ being the second
  node of $s_5$. We then have:
  {\small
  \begin{center}
    \begin{diagram}
      \bullet  & \rTo^{\ \msgbox{\msgbox{\textsf{secret}}{\rho_1\rho_3}}{\rho_1\rho_2}\ } 
      & \preceq & 
      \rTo^{\ \msgbox{\msgbox{\textsf{secret}}{\rho_1\rho_3}}{\rho_2\rho_3}\ } 
      & \bullet\\
      \dImplies&&&&\dImplies\\\\\\
      \bullet       & 
      \lTo^{\ \msgbox{\msgbox{\textsf{newsecret}}{\rho_3\rho_1}}{\rho_2\rho_1}\ }
      & \preceq     & 
      \lTo^{\ \msgbox{\msgbox{\textsf{newsecret}}{\rho_3\rho_1}}{\rho_3\rho_2}\ }
      & \bullet
    \end{diagram}    
  \end{center}
  }
  Above we have actually applied \Did{A1} twice, where the second
  application just added the top $\preceq$.  Note that
  (\ref{skeleton4}) differs from the above because the latter has more
  information about $s_2$ but they are both realized (and DG).

  The set of boxes $B$ is not always empty. For instance, for
  $b=\msgbox{\textsf{secret}}{\rho_1\rho_3}$, with strands
  {\small
  \begin{align*}
    s_2'\ =\ &
    +\msgbox{b}{\rho_1\rho_2}
    \ \Rightarrow\ 
    -\msgbox{b}{\rho_2\rho_1}
    \\
    s_4'\ =\ &
    -{\msgbox{b}{\rho_1\rho_2}}
    \ \Rightarrow\ 
    +\msgbox{b}{\rho_2\rho_3}
    \ \Rightarrow\ 
    -\msgbox{b}{\rho_3\rho_2}
    \ \Rightarrow\ 
    +\msgbox{b}{\rho_2\rho_1}
    \\
    s_5'\ =\ &
    -\msgbox{b}{\rho_2\rho_3}
    \ \Rightarrow\ 
    +\msgbox{b}{\rho_3\rho_2}
  \end{align*}
} and applying \Did{A1} to $s_2'$ with $R=\emptyset$, we get the
following skeleton for $M=b$ and $B=\{\msgbox{b}{\rho_1\rho_2}\}$:
{\small
  \begin{center}
    \begin{diagram}
      \bullet  & \rTo^{\ {\msgbox{b}{\rho_1\rho_2}}\ } 
      & \ & 
      \rTo^{\ \msgbox{b}{\rho_2\rho_3}\ } 
      & \bullet\\
      \dImplies&&&&\dImplies\\\\\\
      \bullet       & 
      \lTo^{\ \msgbox{b}{\rho_2\rho_1}\ }
      & \preceq     & 
      \lTo^{\ \msgbox{b}{\rho_3\rho_2}\ }
      & \bullet
    \end{diagram}    
  \end{center}
  }
\end{example}

\section{A Protocol Description Calculus}
\label{sec:choreography}
We illustrate our ideas with the simplest possible calculus.  The
syntax of this minimal choreography language (based on the Global
Calculus \cite{carbone.honda.yoshida:esop07}) is given by the
following grammar: {\small
  \begin{align*}
    C::=&\phantom{{}\mid\quad{}}\Sigma_i\,\interact{\rho_1}{\rho_2}{op_i}{\tilde M_i}\pfx C_i 
& \text{(interaction)}\\
     &\mid\quad \INACT					              &\text{(inactive)}\\
   \end{align*}} \NI Above, the term
 $\Sigma_i\,\interact{\rho_1}{\rho_2}{op_i}{\tilde M_i}\pfx C_i$
 describes an interaction where a branch with label ${\sf op}_i$ is
 non-deterministically selected and a message $\tilde M_i$ is sent
 from role $\rho_1$ to role $\rho_2$. Each two roles in a choreography
 share a private channel hence it would be redundant to have them
 explicit in the syntax \cite{BCDDDY:concur2008}.  Term $\INACT$
 denotes the inactive system.  Given a choreography $C$, we assume
 that the various $\op {op}$, also on different interactions, are
 distinct: given the lack of an iteration operator e.g.  recursion,
 this is a constraint that can be imposed statically and we include in
 the well-formedness condition at the end of this section.

Our mini-language can be equipped with a standard trace semantics with
configurations $C\lts{\mu}{} C'$ where $\mu =
(\rho_1,\rho_2,\textsf{op}_i,\tilde M_i)$ contains the parameters of
the interaction performed i.e. $\interact{\rho_1}{\rho_2}{op_i}{\tilde
  M_i}\pfx C_i\lts{ (\rho_1,\rho_2,\textsf{op}_i,\tilde M_i)}{}{C_i}$. A
sequence of labels $\{\mu_i\}_i$ describes the temporal order in which
the various described communications take place and it is called {\em
  trace}.

 \begin{assumption}[Well-Formedness] A choreography $C$ is well-formed
   whenever:
   \begin{itemize}
   \item All $\op{op}$'s are distinct;

   \item let $\Gamma$ be a set of pairs $\rho:\tilde M$. Then,
     $\Gamma\vdash C$ such that for all $\rho$, $\Gamma(\rho)$ has no
     boxes and $\vdash$ is defined by the following rules:
     {\small 
     \begin{align*}
       \Did{T-Interact}\
       &
       \Rule
       {\tilde M_i\subseteq\Gamma(\rho_1)\qquad\Gamma[\rho_2\mapsto\Gamma(\rho_2)\cup\{\tilde M_i\}]\vdash
         C_i\qquad \rho_2\in{\sf top}(C_i)}
       {\Gamma\vdash\Sigma_i\,\interact{\rho_1}{\rho_2}{op_i}{\tilde
           M_i}\pfx C_i}
       &
       \Did{T-Inact}\
       &
       \Rule
       {}
       {\Gamma\vdash\INACT}
       \end{align*}
       \begin{align*}
       \Did{T-Box$_1$}\
       &
       \Rule
       {\Gamma\vdash C \qquad \tilde M\in\Gamma(\rho_1)}
       {\Gamma[\rho_1\mapsto\Gamma(\rho_1)\cup\{\msgbox{\tilde M}{\rho_1\rho_2}\}]\vdash C}
       &
       \Did{T-Box$_2$}\
       &
       \Rule
       {\Gamma\vdash C\qquad\msgbox{\tilde M}{\rho_1\rho_2}\in\Gamma(\rho_2)}
       {\Gamma[\rho_2\mapsto\Gamma(\rho_2)\cup\{\tilde M\}]\vdash C}
     \end{align*}}
     where $\TOP{\Sigma_i\,\interact{\rho_1}{\rho_2}{op_i}{\tilde
         M_i}\pfx C_i}=\{\rho_1\}$ and $\TOP{\INACT}=\mathcal R$.
   \end{itemize}
 \end{assumption}
 The rules above are a simple static check for ensuring that a box
 $\msgbox{\tilde M}{\rho_1\rho_2}$ always originate by an interaction
 from $\rho_1$ and can only be opened by $\rho_2$ upon reception of
 the box (maybe nested in other boxes). An environment $\Gamma$ is a
 function that associates a set of messages to a role.
 \Did{T-Interact} checks $\Gamma(\rho_1)$ contains each $\tilde M_i$
 and allows $\rho_2$ to use $\tilde M_i$ in $C_i$. Moreover, the rules
 checks that $\rho_2$ is the sender in $C_i$. \Did{T-Box$_1$} says
 that if $\rho_1$ knows $\tilde M$ then it can also create $\msgbox
 {\tilde M}{\rho_1\rho_2}$ for any $\rho'$. Dually, in
 \Did{T-Box$_2$}, if $\rho_2$ knows $\msgbox{\tilde M}{\rho_1\rho_2}$
 then it also knows $\tilde M$. Rule \Did{T-Inact} allows to type
 $\INACT$ with any $\Gamma$. 

\begin{example}[Buyer-Seller Protocol]\rm\label{example:buyersellerprotocol}
  Hereby, we report a Buyer-Seller financial protocol
  \cite{carbone.honda.yoshida:esop07,CG09}.  A buyer \textsf{Buyer}
  asks a seller \textsf{Seller} for a quote about a product. If the
  quote is accepted, \textsf{Buyer} will send its credit card
  \texttt{card} together with the accepted quote to \textsf{Seller}
  who will forward it to a bank \textsf{Bank}. The bank will check if
  payment can be done and, if so, reply with a receipt
  \texttt{receipt} which will be forwarded to \textsf{Buyer} by
  \textsf{Seller}.  In our mini-language we use boxes to make sure
  that the credit card number can only be read by \textsf{Bank} and
  that \textsf{Seller} does not change the accepted quote: {\small
\begin{align*}
  1.\quad &
  \interact{\textsf{Buyer}}{\textsf{Seller}}{\textsf{Req}}{\texttt{prod}}\pfx\
  \interact{\textsf{Seller}}{\textsf{Buyer}}{\textsf{Reply}}{\texttt{quote}}\pfx\\
  2.\quad & (\ \interact{\textsf{Buyer}}{\textsf{Seller}}{\textsf{Accept}}
  {\msgbox{(\texttt{quote},\texttt{card})}{\textsf{Buyer}\textsf{Bank}}}\pfx\
  \interact{\textsf{Seller}}{\textsf{Bank}}{\textsf{Pay}}
  {(\texttt{quote},\msgbox{(\texttt{quote},\texttt{card})}{\textsf{Buyer}\textsf{Bank}})}\pfx\\
  3.\quad & \phantom{(\ }
  (\ \interact{\textsf{Bank}}{\textsf{Seller}}{\textsf{Ok}}
  {\msgbox{\texttt{receipt}}{\textsf{Bank}\textsf{Buyer}}}\pfx
  \interact{\textsf{Seller}}{\textsf{Buyer}}{\textsf{Succ}}
  {\msgbox{\texttt{receipt}}{\textsf{Bank}\textsf{Buyer}}}\\
  4.\quad & \phantom{(\
    \interact{\textsf{Bank}}{\textsf{Seller}}{\textsf{Pay}}
    {\texttt{card}}\pfx\  (\ } \qquad\qquad\qquad\qquad+\\
  5.\quad & \phantom{(\ (\ }
  \interact{\textsf{Bank}}{\textsf{Seller}}{\textsf{NotOk}}
  {\textsf{reason}}\pfx
  \interact{\textsf{Seller}}{\textsf{Buyer}}{\textsf{Fail}}{\textsf{reason}}\ )\\
  6.\quad & \qquad\qquad\qquad\qquad\qquad\qquad+\\
  7.\quad & \phantom{(\ }
  \interact{\textsf{Buyer}}{\textsf{Seller}}{\textsf{Reject}}{})
\end{align*}} \NI Line 1. denotes the quote request and reply. Lines
2. and 7. are computational branches corresponding to acceptance and
rejection of the quote respectively. If the quote is accepted,
\textsf{Buyer} will send its credit card in the box
$\msgbox{\textsf{quote},\textsf{card}}{\textsf{Buyer},\textsf{Bank}}$
meaning that \textsf{Seller} cannot see it. The box is then forwarded
to \textsf{Bank} together with the quote offered by \textsf{Seller}
who checks that everything is fine (line 2.). If the transaction can
be finalised, a receipt is forwarded to \textsf{Buyer}.  Otherwise, a
\textsf{NotOK} message will be delivered.  \textsf{Bank} boxes the
receipt so that it cannot be seen or changed by \textsf{Seller}.

\end{example}

\subsection{Abstract Strand Semantics}
The {\em abstract strand} semantics (AS semantics) is the minimum
function $\absem{-}:C\rightarrow 2^S\times(\mathcal R\rightarrow2^S)$
(for $S$ a set of strands) satisfying the rules in
Table~\ref{table:abstractsemantics}. The function inputs a
choreography and returns a set of strands paired with a function that
maps strands into a role $\rho$ in $\mathcal R$ (all the possible runs
for $\rho$).  
 \begin{table}[t]
{\small
\begin{displaymath}
  \begin{array}{rl}
    \Did{AS-Com}\ 
    &
    \Rule
    {
      \begin{array}{l}
        \absem{C_i}={(S_i,\who_i)}
      \end{array}
    }
    {
      \absem{\Sigma_i\,\interact{\rho_1}{\rho_2}{op_i}{\tilde M_i}\pfx C_i}
      =
      \bigcup_i
      \left(
        \begin{array}{l}
          \op{extend}(S_i,\msgbox{(\op{op_i},\tilde M_i)}{\rho_1\rho_2},\rho_1,\rho_2,\who_i)
        \end{array}
      \right)
    }
  \end{array}
\end{displaymath}\\
\begin{displaymath}
  \begin{array}{rl}
  \Did{AS-Zero}\ 
  &
  \Rule
  {
  }
  {
    \absem{\INACT}={(\{+\bullet^\rho\}_\rho,\lambda \rho\pfx\{+\bullet^\rho\})}
  }
  \end{array}
\end{displaymath}} 
\caption{Abstract Strand Semantics for Choreography}
\label{table:abstractsemantics}
\end{table}
These strands are templates, and we may use substitutions to ``plug
in'' alternate values for the parameters in the choreography.  Since
these parameters do not include the labels $\mathsf{op}_i$, we define: 

A substitution $\sigma$ is a \emph{parameter substitution} if for
every label $\mathsf{op}_i$, $\sigma(\mathsf{op}_i)=\mathsf{op}_i$.
The \emph{strand space of a choreography} $C$ is the strand space
generated by applying parameter substitutions to $\absem{C}$.  We say
that a skeleton $\skel$ is \emph{over} $\absem{C}$ if all of its
strands belong to this strand space.

Rule \Did{AS-Zero} gives semantics to the inactive choreography
$\INACT$ by creating a strand $+\bullet^\rho$ for each role
$\rho\in\mathcal R$. Rule \Did{AS-Com} gives the semantics to the term
(interaction) of a choreography. The idea is to prefix, for every
branch, every strand of $\rho_1$ with $+\msgbox{(\op{op_i},\tilde
  M_i)}{\rho_1\rho_2}$ and every strand of $\rho_2$ with
$-\msgbox{(\op{op_i},\tilde M_i)}{\rho_1\rho_2}$ where, in general,
$(\op{op},\tilde M)$ denotes the vector $(\op{op},M_{1}, \ldots,
M_{k})$. The main part is played by the function $\op{extend}$ hereby
defined as:
\begin{displaymath}
  \op{extend}(S,M,\rho_1,\rho_2,\who)
  = 
  \left(
  \begin{array}{l}
    S \backslash(\who(\rho_1)\cup\who(\rho_2))\quad\cup\\
    \left\{a\Rightarrow s\ \mid\ 
    \begin{array}{l}
      (s\in\who(\rho_1)\land a=+M)\quad\lor \\
      (s\in\who(\rho_2)\land a=-M)
    \end{array}
    \right\}
  \end{array}
\right)
\end{displaymath}
The above definition says that we include all those strands which are
not in $\who(\rho_1)$ and in $\who(\rho_2)$. Then, we must prefix all
those strands in $\who(\rho_1)$ with node $+M$ and all those strands
in $\who(\rho_2)$ with $-M$.  For well-formed choreographies, we have
the following:
\begin{proposition}\label{proposition:strandconsistency}
  Let $C$ be a well-formed choreography and $(S,\who{})$ its
  semantics. Then each message $\msgbox{\tilde M}{\rho_1\rho_2}$
  always originates in $\who(\rho_1)$ and can only be opened in
  $\who(\rho_2)$. 
\end{proposition}

\smallskip 

\begin{example}[Semantics of the Buyer-Seller Protocol]\rm
  Unlike in \cite{CG09}, because of the presence of corrupted roles
  (and participants), we cannot give the semantics of a choreography
  describing a security protocol simply by giving a set of
  executions. Therefore, the semantics of the buyer-seller protocol is
  a set of strands from which we would like to build the possible
  executions depending on which roles are compromised.  Given the
  choreography in Example~\ref{example:buyersellerprotocol}, we get
  the following strands: {\small
  \begin{align*}
    a)\quad &
    +\msgbox{(\op{Req},\op{prod})}{\textsf{B}\textsf{S}}\Rightarrow
    -\msgbox{(\textsf{Reply},\texttt{quote})}{\textsf{S}\textsf{B}}\Rightarrow
    +\msgbox{(\textsf{Accept},\msgbox{(\texttt{quote},\texttt{card})}{\textsf{B}\textsf{Bk}})}{\textsf{B}\textsf{S}}\Rightarrow
    -\msgbox{(\textsf{Succ},\msgbox{\texttt{receipt}}{\textsf{Bk}\textsf{B}})}{\textsf{S}\textsf{B}}
    \\
    b)\quad &
    +\msgbox{(\op{Req},\op{prod})}{\textsf{B}\textsf{S}}\Rightarrow
    -\msgbox{(\textsf{Reply},\texttt{quote})}{\textsf{S}\textsf{B}}\Rightarrow
    +\msgbox{(\textsf{Accept},\msgbox{(\texttt{quote},\texttt{card})}{\textsf{B}\textsf{Bk}})}{\textsf{B}\textsf{S}}\Rightarrow
    -\msgbox{(\textsf{Fail},\texttt{reason})}{\textsf{S}\textsf{B}}
    \\
    c)\quad &
    +\msgbox{(\op{Req},\op{prod})}{\textsf{B}\textsf{S}}\Rightarrow
    -\msgbox{(\textsf{Reply},\texttt{quote})}{\textsf{S}\textsf{B}}\Rightarrow
    +\msgbox{\textsf{Reject}}{\textsf{B}\textsf{S}}
    \\
    d)\quad & 
    -\msgbox{(\op{Req},\op{prod})}{\textsf{B}\textsf{S}}\Rightarrow
    +\msgbox{(\textsf{Reply},\texttt{quote})}{\textsf{S}\textsf{B}}\Rightarrow
    -\msgbox{(\textsf{Accept},\msgbox{(\texttt{quote},\texttt{card})}{\textsf{B}\textsf{Bk}})}{\textsf{B}\textsf{S}}\Rightarrow    
    \\ &
    \Rightarrow+\msgbox{(\textsf{Pay},\texttt{quote},\msgbox{(\texttt{quote},\texttt{card})}{\textsf{B}\textsf{Bk}})}{\textsf{S}\textsf{Bk}}\Rightarrow
    -\msgbox{(\textsf{Ok},\msgbox{\texttt{receipt}}{\textsf{Bk}\textsf{B}})}{\textsf{Bk}\textsf{S}}\Rightarrow
    +\msgbox{(\textsf{Succ},\msgbox{\texttt{receipt}}{\textsf{Bk}\textsf{B}})}{\textsf{S}\textsf{B}}
    \\
    e)\quad & 
    -\msgbox{(\op{Req},\op{prod})}{\textsf{B}\textsf{S}}\Rightarrow
    +\msgbox{(\textsf{Reply},\texttt{quote})}{\textsf{S}\textsf{B}}\Rightarrow
    -\msgbox{(\textsf{Accept},\msgbox{(\texttt{quote},\texttt{card})}{\textsf{B}\textsf{Bk}})}{\textsf{B}\textsf{S}}\Rightarrow
    \\ &
    \Rightarrow+\msgbox{(\textsf{Pay},\texttt{quote},\msgbox{(\texttt{quote},\texttt{card})}{\textsf{B}\textsf{Bk}})}{\textsf{S}\textsf{Bk}}\Rightarrow
    -\msgbox{(\textsf{NotOk},\texttt{reason})}{\textsf{Bk}\textsf{S}}\Rightarrow
    +\msgbox{(\textsf{Fail},\texttt{reason})}{\textsf{S}\textsf{B}}
    \\
    f)\quad & 
    -\msgbox{(\op{Req},\op{prod})}{\textsf{B}\textsf{S}}\Rightarrow
    +\msgbox{(\textsf{Reply},\texttt{quote})}{\textsf{S}\textsf{B}}\Rightarrow
    -\msgbox{\textsf{Reject}}{\textsf{B}\textsf{S}}
    \\
    g)\quad &
    -\msgbox{(\textsf{Pay},\texttt{quote},\msgbox{(\texttt{quote},\texttt{card})}{\textsf{B}\textsf{Bk}})}{\textsf{S}\textsf{Bk}}\Rightarrow
    +\msgbox{(\textsf{Ok},\msgbox{\texttt{receipt}}{\textsf{Bk}\textsf{B}})}{\textsf{Bk}\textsf{S}}
    \\
    h)\quad &
    -\msgbox{(\textsf{Pay},\texttt{quote},\msgbox{(\texttt{quote},\texttt{card})}{\textsf{B}\textsf{Bk}})}{\textsf{S}\textsf{Bk}}\Rightarrow
    +\msgbox{(\textsf{NotOk},\texttt{reason})}{\textsf{Bk}\textsf{S}}
  \end{align*}} where $\textsf{B}$ is the buyer, $\textsf{S}$ is the
seller and $\textsf{Bk}$ is the bank. Above, strands $a)$, $b)$ and
$c)$ belong to \textsf{B} while $d)$, $e)$ and $f)$ belong to
\textsf{S}. Strands $g)$ and $h)$ are instead the local behaviour of
\textsf{Bk}.
\end{example}

\subsection{Realized Skeletons for Choreography}
We now apply the theory developed in the previous section to abstract
spaces which are in fact the semantics of a choreography.

In the sequel we say that $\skel$ is {\em over} $\absem{C}$ whenever
it is obtained from the regular, non compromised strands in
$\absem{C}$.  The following result states that whenever \Did{A1} is
not applicable, we have reached a realized skeleton.


\begin{lemma}[Realized Skeletons]\label{lemma:soundness:realized}
  Let $C$ be a well-formed choreography and let $\skel$ be a skeleton
  over $\absem{C}$ such that \Did{A1} is not applicable. Then $\skel$
  is realized.
\end{lemma}
\begin{proof}

  By Proposition~\ref{theorem:cutrealized}, $\skel'$ is realized if
  and only if all its cuts are solved. Let us assume, by
  contradiction, that $\textsf{Cut}(\msgbox M{\rho_1\rho_2}, B,
  \skel)$ is unsolved for some $\msgbox M{\rho_1\rho_2}$ and $B$.

  By Proposition~\ref{theorem:minimalB}, we know that also
  $\textsf{Cut}(\msgbox M{\rho_1\rho_2}, B', \skel)$ is unsolved for
  $B'=\{b \pp n'\prec_\skel n\;\text{ s.t. }\ \msgbox{\tilde
    M}{\rho_1\rho_2} \sqsubset b\sqsubseteq \textsf{msg}(n') \land
  \textsf{rcv}(b)\not\in R \ \}$ for some $\preceq_{\skel}$-minimal
  input node $n$ in $\textsf{Cut}(\msgbox{M}{\rho_1\rho_2},B,\mathbb
  A)$ and $\rho_1\not\in R$. As a consequence, we also have that
  $\msgbox{\tilde M}{\rho_1\rho_2} \dagger^{B'}\textsf{msg}(n)$.

  Now, if we prove the existence of some positive node $m\not\in\skel$
  such that $\forall m'\,.\, m'\prec_{\mathbb A}n \lor m'\Rightarrow^+
  m$ implies $\msgbox{\tilde
    M}{\rho_1\rho_2}\odot^{B'}\textsf{msg}(m')$ where $\msgbox{\tilde
    M}{\rho_1\rho_2}\dagger^{B'}\textsf{msg}(m)$ and $m\not\in R$ then
  we can apply \Did{A1} to $\skel$ hence having a contradiction. We
  distinguish two cases:
  \begin{itemize}

  \item $B'=\emptyset$. In this case, the unsolved cut is saying that
    we must explain where the box $c$ has been created. As $\rho_1$ is
    not compromised, we must add a node belonging to $\rho_1$ sending
    $c$. The existence of such a node is ensured by well-formedness.

  \item $B'\not=\emptyset$. As $B$ is non-empty, then we must explain
    how $c$ has come out of some message box $\msgbox{\tilde
      M'}{\rho_3\rho_4}$ in $B$. But if that is the case, as $\rho_4$
    is not compromised, a node belonging to $\rho_4$ must have
    performed such operation. The existence of such a node is ensured
    by well-formedness.

  \end{itemize}
  \NI Note that in both cases above, we are exploiting the fact that
  the two well-formedness conditions impose that the operations for
  creation and opening of a box are performed consistently on the same
  choreography branches i.e. role strands.
\end{proof}

The following result states that whenever \Did{A2} is not applicable
to $\skel$ then $\skel$ is DG. 
\begin{lemma}\label{lemma:soundness:DG}
  Let $C$ be a well-formed choreography and let $\skel$ be a skeleton
  over $\absem{C}$ such that \Did{A2} is not applicable. Then $\skel$
  is DG.
\end{lemma}
\begin{proof}
  If that is not the case then, by definition of delivery guaranteed
  skeleton, we would be able to apply \Did{A2}. This is simply because
  whenever we add a positive node $n$ to $\skel$ we always have
  another strand belonging to a different role and containing a
  negative node $m$ such that $\msg{m}=\msg n$.
\end{proof}

We finally have the following two results:
\begin{theorem}[Soundness]
\label{thm:soundness}
  Let $\skel$ be a single-stranded skeleton over $\absem C$ and let
  $\skel\ltsstar S\skel'\not\rightarrow$. Then, $\skel'$ is a DG
  shape.
\end{theorem}
\begin{proof}
  By the previous lemmas, we know that $\skel'$ is realized and DG. We
  must prove that there exists a homomorphism $H:\skel\mapsto\skel'$
  which is a shape.

  As $\skel\ltsstar S\skel'$ then, by
  Proposition~\ref{proposition:reductionaugmentation}, we can choose
  $H=H_k\circ\ldots\circ H_0$ where $H_i:\skel_i\mapsto\skel_{i+1}$
  for for $\skel_0=\skel$ and $\skel_{k+1}=\skel'$ and some
  $\skel_i$. We shall prove that $H_k\circ\ldots\circ
  H_i:\skel_i\mapsto \skel'$ is a shape for $\skel_i$ for all $i$. We
  do it by induction on $j=k-i$.

  \begin{itemize}

  \item {\bf Base Case.} $j=1$. We have to prove that
    $H_k:\skel_k\mapsto\skel'$ is a shape for $\skel_k$. By
    Proposition~\ref{proposition:reducesunrealized}, we know that
    $\skel_k$ is not realized and/or not DG. Hence, $H_k$ must be the
    minimum homomorphism mapping $\skel_k$ to a DG realized
    skeleton. In fact, both \Did{A1} and \Did{A2}, add the minimum
    node explaining a box or receiving a pending output.

  \item {\bf Inductive Case.}  Let us assume that $j=i+1$. By
    induction hypothesis we know that $H_{k}\circ\ldots\circ
    H_{i+1}:\skel_{i+1}\mapsto \skel'$ is a shape for
    $\skel_{i+1}$. But then, as augmentations are minimal strictly
    monotone embedding with respect to shapes, we have that also
    $H_{k}\circ\ldots\circ H_{i}:\skel_{i}\mapsto \skel'$ is a shape
    for $\skel_{i}$.

  \end{itemize}

\end{proof}
\begin{theorem}[Termination]
  \label{thm:termination}
  Let $\skel$ be a single-stranded skeleton over $\absem C$. Then, we
  can reduce $\skel$ only a finite number of times.
\end{theorem}
\begin{proof}
  $\absem C$ is finite and the reduction rules are augmentation
  (increase the number of nodes). As the same node cannot be added
  twice, we must eventually exhaust all nodes.
\end{proof}

\begin{example}[Shapes of the Buyer-Seller Protocol]\rm
  We show how to compute some shapes of the Buyer-Seller protocol
  starting from its semantics given in the previous section. We start
  from the buyer's strand $a)$ assuming that seller is
  compromised. Applying \Did{A1} to its fourth node, we get:
{\small  
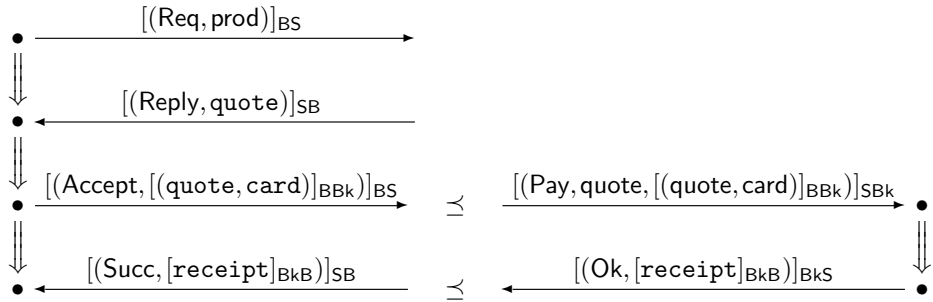
  \begin{center}
    \begin{diagram}
      \bullet 
      & 
      \rTo^{\ \msgbox{(\op{Req},\op{prod})}{\textsf{B}\textsf{S}}\ } 
      &
      \quad\ \quad
      &
      &
      \\\\
      \dImplies
      & 
      \
      &
      \
      &
      \
      &
      \\\\
      \bullet 
      & 
      \lTo^{\ \msgbox{(\textsf{Reply},\texttt{quote})}{\textsf{S}\textsf{B}}\ } 
      &
      \quad\ \quad
      &
      &
      \\\\
      \dImplies
      & 
      \
      &
      \
      &
      \
      &
      \\\\
      \bullet 
      & 
      \rTo^{\ \msgbox{(\textsf{Accept},\msgbox{(\texttt{quote},\texttt{card})}{\textsf{B}\textsf{Bk}})}{\textsf{B}\textsf{S}}\ } 
      &
      \quad\preceq \quad
      &
      \rTo^{\ \msgbox{(\op{Pay},\textsf{quote},\msgbox{(\textsf{quote},\textsf{card})}{\textsf{B}\textsf{Bk}})}{\textsf{S}\textsf{Bk}}\ } 
      &
      \bullet
      \\\\
      \dImplies
      & 
      \
      &
      \
      &
      \
      &
      \dImplies
      \\\\
      \bullet 
      & 
      \lTo^{\ \msgbox{(\textsf{Succ},\msgbox{\texttt{receipt}}{\textsf{Bk}\textsf{B}})}{\textsf{S}\textsf{B}}\ } 
      &
      \quad\preceq\quad
      &
      \lTo^{\ \msgbox{(\textsf{Ok},\msgbox{\texttt{receipt}}{\textsf{Bk}\textsf{B}})}{\textsf{Bk}\textsf{S}}\ } 
      &
      \bullet
    \end{diagram}
  \end{center}
} Note that, we have actually applied \Did{A1} twice: the second time
it was applied to the first node of the new strand and its result was
only adding the relation $\preceq$.  The image of the shape for strand
$b)$, the case when the bank does not accept the transaction, is
similar. Let us now consider $d)$ and let us assume that buyer is
compromised. In this case, for
$M=\msgbox{(\textsf{Pay},\texttt{quote}\msgbox{(\texttt{quote},\texttt{card})}{\textsf{B}\textsf{Bk}})}{\textsf{S}\textsf{Bk}}$,
by applying \Did{A1} (twice) we get: {\small
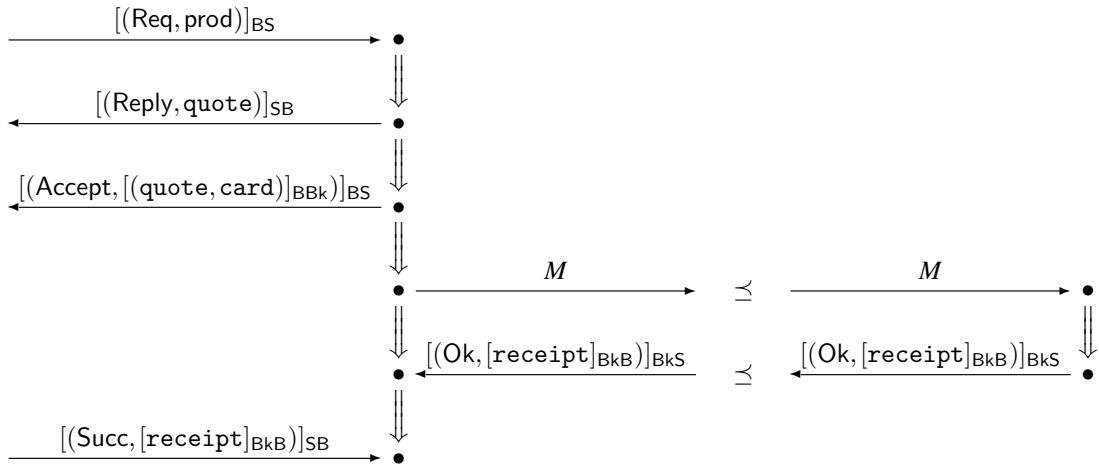
  \begin{center}
    \begin{diagram}
      \
      &
      \rTo^{\ \msgbox{(\op{Req},\op{prod})}{\textsf{B}\textsf{S}}\ } 
      &
      \bullet 
      & 
      \
      &
      \quad\ \quad
      &
      &
      \\\\
      \
      &
      \
      &
      \dImplies
      & 
      \
      &
      \
      &
      \
      &
      \\\\
      \
      &
      \lTo^{\ \msgbox{(\textsf{Reply},\texttt{quote})}{\textsf{S}\textsf{B}}\ } 
      &
      \bullet 
      & 
      \
      &
      \quad\ \quad
      &
      &
      \\\\
      \
      &
      \
      &
      \dImplies
      & 
      \
      &
      \
      &
      \
      &
      \\\\
      \
      &
      \lTo^{\ \msgbox{(\textsf{Accept},\msgbox{(\texttt{quote},\texttt{card})}{\textsf{B}\textsf{Bk}})}{\textsf{B}\textsf{S}}\ }
      &
      \bullet 
      \\\\
      \
      &
      \
      &
      \dImplies
      \\\\
      \
      &
      \
      &
      \bullet 
      & 
      \rTo^{\ M\ }
      &
      \quad\preceq \quad
      &
      \rTo^{\ M\ }
      &
      \bullet
      \\\\
      \
      &
      \
      &
      \dImplies
      & 
      \
      &
      \
      &
      \
      &
      \dImplies
      \\\\
      \
      &
      \
      &
      \bullet 
      & 
      \lTo^{\ \msgbox{(\textsf{Ok},\msgbox{\texttt{receipt}}{\textsf{Bk}\textsf{B}})}{\textsf{Bk}\textsf{S}}\ }
      &
      \quad\preceq\quad
      &
      \lTo^{\ \msgbox{(\textsf{Ok},\msgbox{\texttt{receipt}}{\textsf{Bk}\textsf{B}})}{\textsf{Bk}\textsf{S}}\ }
      &
      \bullet
      \\\\
      \
      &
      \
      &
      \dImplies
      \\\\
      \
      &
      \rTo^{\ \msgbox{(\textsf{Succ},\msgbox{\texttt{receipt}}{\textsf{Bk}\textsf{B}})}{\textsf{S}\textsf{B}}\ }
      &
      \bullet 
    \end{diagram}
  \end{center}
}
\end{example}

\begin{example}\rm
  Let us consider a slightly different version of the Buyer-Seller
  protocol, where the buyer does not include the quote together with
  her credit card. In particular we would have the new following
  strands (the missing ones are unchanged): 
{\small
  \begin{align*}
    a')\quad &
    +\msgbox{(\op{Req},\op{prod})}{\textsf{B}\textsf{S}}\Rightarrow
    -\msgbox{(\textsf{Reply},\texttt{quote})}{\textsf{S}\textsf{B}}\Rightarrow
    +\msgbox{(\textsf{Accept},\msgbox{\texttt{card}}{\textsf{B}\textsf{Bk}})}{\textsf{B}\textsf{S}}\Rightarrow
    -\msgbox{(\textsf{Succ},\msgbox{\texttt{receipt}}{\textsf{Bk}\textsf{B}})}{\textsf{S}\textsf{B}}
    \\
    b')\quad &
    +\msgbox{(\op{Req},\op{prod})}{\textsf{B}\textsf{S}}\Rightarrow
    -\msgbox{(\textsf{Reply},\texttt{quote})}{\textsf{S}\textsf{B}}\Rightarrow
    +\msgbox{(\textsf{Accept},\msgbox{\texttt{card}}{\textsf{B}\textsf{Bk}})}{\textsf{B}\textsf{S}}\Rightarrow
    -\msgbox{(\textsf{Fail},\texttt{reason})}{\textsf{S}\textsf{B}}
    \\
    d')\quad & 
    -\msgbox{(\op{Req},\op{prod})}{\textsf{B}\textsf{S}}\Rightarrow
    +\msgbox{(\textsf{Reply},\texttt{quote})}{\textsf{S}\textsf{B}}\Rightarrow
    -\msgbox{(\textsf{Accept},\msgbox{\texttt{card}}{\textsf{B}\textsf{Bk}})}{\textsf{B}\textsf{S}}\Rightarrow    
    \\ &
    \Rightarrow+\msgbox{(\textsf{Pay},\texttt{quote},\msgbox{\texttt{card}}{\textsf{B}\textsf{Bk}})}{\textsf{S}\textsf{Bk}}\Rightarrow
    -\msgbox{(\textsf{Ok},\msgbox{\texttt{receipt}}{\textsf{Bk}\textsf{B}})}{\textsf{Bk}\textsf{S}}\Rightarrow
    +\msgbox{(\textsf{Succ},\msgbox{\texttt{receipt}}{\textsf{Bk}\textsf{B}})}{\textsf{S}\textsf{B}}
    \\
    e')\quad & 
    -\msgbox{(\op{Req},\op{prod})}{\textsf{B}\textsf{S}}\Rightarrow
    +\msgbox{(\textsf{Reply},\texttt{quote})}{\textsf{S}\textsf{B}}\Rightarrow
    -\msgbox{(\textsf{Accept},\msgbox{\texttt{card}}{\textsf{B}\textsf{Bk}})}{\textsf{B}\textsf{S}}\Rightarrow
    \\ &
    \Rightarrow+\msgbox{(\textsf{Pay},\texttt{quote},\msgbox{\texttt{card}}{\textsf{B}\textsf{Bk}})}{\textsf{S}\textsf{Bk}}\Rightarrow
    -\msgbox{(\textsf{NotOk},\texttt{reason})}{\textsf{Bk}\textsf{S}}\Rightarrow
    +\msgbox{(\textsf{Fail},\texttt{reason})}{\textsf{S}\textsf{B}}
    \\
    g')\quad &
    -\msgbox{(\textsf{Pay},\texttt{quote},\msgbox{\texttt{card}}{\textsf{B}\textsf{Bk}})}{\textsf{S}\textsf{Bk}}\Rightarrow
    +\msgbox{(\textsf{Ok},\msgbox{\texttt{receipt}}{\textsf{Bk}\textsf{B}})}{\textsf{Bk}\textsf{S}}
    \\
    h')\quad &
    -\msgbox{(\textsf{Pay},\texttt{quote},\msgbox{\texttt{card}}{\textsf{B}\textsf{Bk}})}{\textsf{S}\textsf{Bk}}\Rightarrow
    +\msgbox{(\textsf{NotOk},\texttt{reason})}{\textsf{Bk}\textsf{S}}
  \end{align*}} If the seller is corrupted, starting from $g')$ and
applying \Did{A1} to its first node, we get the realized skeleton:
{\small
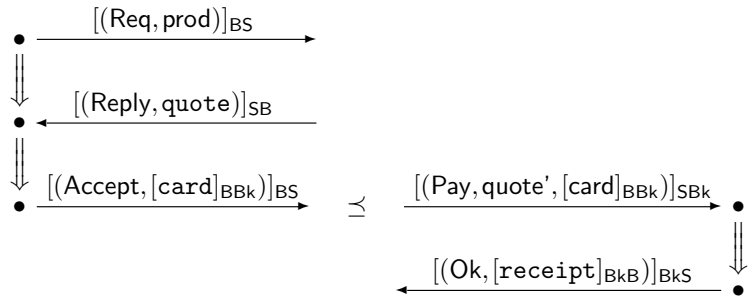
  \begin{center}
    \begin{diagram}
      \bullet 
      & 
      \rTo^{\ \msgbox{(\op{Req},\op{prod})}{\textsf{B}\textsf{S}}\ } 
      &
      \quad\ \quad
      &
      &
      \\\\
      \dImplies
      & 
      \
      &
      \
      &
      \
      &
      \\\\
      \bullet 
      & 
      \lTo^{\ \msgbox{(\textsf{Reply},\texttt{quote})}{\textsf{S}\textsf{B}}\ } 
      &
      \quad\ \quad
      &
      &
      \\\\
      \dImplies
      & 
      \
      &
      \
      &
      \
      &
      \\\\
      \bullet 
      & 
      \rTo^{\ \msgbox{(\textsf{Accept},\msgbox{\texttt{card}}{\textsf{B}\textsf{Bk}})}{\textsf{B}\textsf{S}}\ } 
      &
      \quad\preceq \quad
      &
      \rTo^{\ \msgbox{(\op{Pay},\textsf{quote'},\msgbox{\textsf{card}}{\textsf{B}\textsf{Bk}})}{\textsf{S}\textsf{Bk}}\ } 
      &
      \bullet
      \\\\
      \
      & 
      \
      &
      \
      &
      \
      &
      \dImplies
      \\\\
      \
      & 
      \
      &
      \quad\ \quad
      &
      \lTo^{\ \msgbox{(\textsf{Ok},\msgbox{\texttt{receipt}}{\textsf{Bk}\textsf{B}})}{\textsf{Bk}\textsf{S}}\ } 
      &
      \bullet
    \end{diagram}
  \end{center}}
  The realized skeleton above shows a flaw, or at least an undesirable
  aspect of this version of the protocol.  The value \textsf{quote}
  that the client accepted can be different from \textsf{quote'}
  received by the bank, allowing for the seller to cheat on the quote
  agreed with the buyer.
\end{example}
%


\section{Conclusions}
\label{sec:conclusions}
In this paper, we have used the strand space framework to study the
possible behaviors of choreographies executing in the presence of
compromised principals.  In this framework, the strands of the
uncompromised regular participants can freely interact with each other
and with behaviors possible for corrupted parties.  We clarified these
behaviors by presenting a pair of transition rules which generate all
of the minimal, essentially different executions.  

It is a strength of this approach that it allows us to formulate and
characterize a number of interesting properties.  For instance, what
about the relationship between shapes (namely minimal executions) and
other, possibly non-minimal executions?  One might expect that
non-minimal executions would be disjoint unions of copies of shapes.
However, this intuition requires a property of choreographies, which
may be characterized syntactically.  In effect, it requires that when
the choreography has a choice, then the same principals are active
across both branches of the choice (except possibly the last principal
on one branch).  This corresponds to an assumption of~\cite{CDFBL09}.
We also conjecture that, under these assumptions, shapes are {\em
run-once} i.e. they are such that there is at most one strand
belonging to each role.  In future work we intend to explore
properties of this kind, in particular when the choreography language
is extended with parallel composition and recursive behaviour.

We also intend to study the relation between protocol descriptions at
the choreography-and-box level and at the concrete cryptographic
level.  We intend to investigate properties of protocol
transformations in general~\cite{Guttman09a} in order to develop
fine-grained principles governing how to generate cryptographic
implementations for choreographies requiring security infrastructures.

\bibliographystyle{eptcs} 
\bibliography{session}

\end{document}